\documentclass[a4paper,10pt]{article}
\usepackage{fullpage}
\usepackage[dvipdfmx]{graphicx}
\usepackage[T1]{fontenc}
\usepackage{lineno}
\usepackage{amsmath,amssymb,mathtools,amsthm,amsfonts}
\usepackage{url}
\usepackage{comment}
\usepackage{fancybox,ascmac}
\usepackage[subrefformat=parens]{subcaption}
\usepackage{color}
\usepackage{paralist}
\usepackage{comment}
\usepackage{authblk}
\usepackage{enumerate}
\usepackage{thm-restate}
\usepackage[colorlinks=true,linkcolor=blue,citecolor=blue]{hyperref}
\usepackage{url}
\usepackage{multirow}%
\usepackage{mathrsfs}%
\usepackage{xcolor}%
\usepackage[ruled,vlined,linesnumbered]{algorithm2e}
\usepackage{algpseudocode}
\usepackage{lineno}

\newtheorem{theorem}{Theorem}
\newtheorem{lemma}{Lemma}

\newtheorem{corollary}{Corollary}

\newtheorem{definition}{Definition}

\newtheorem{claim}{Claim}
\long\def\nop#1{}

\begin{document}
\title{The Complexity of Boolean Connectivity Problem of $k$-Horn Formulas}
\author[1]{Takashi Horiyama}
\author[1]{Shoon Mineyoshi}
\author[1]{Yuto Okura}
\author[1]{Kazuhisa Seto}
\author[2]{Junichi Teruyama}
\affil[1]{Hokkaido University\\
\texttt{\{horiyama,seto\}@ist.hokudai.ac.jp}}
\affil[2]{University of Hyogo\\
\texttt{junichi.teruyama@gsis.u-hyogo.ac.jp}}
\date{}

\maketitle

\abstract{The Boolean connectivity problem asks whether the set of satisfying assignments of a given Boolean formula forms a connected subgraph in the $n$-dimensional hypercube. 
This problem is known to be $\mathsf{coNP}$-complete, even when restricted to $k$-Horn formulas for $k \geq 3$, as shown by Makino, Tamaki, and Yamamoto. 
In this paper, we further investigate the computational complexity of {\sc Conn $k$-Horn}, the Boolean connectivity problem for $k$-Horn formulas.
We provide algorithmic and hardness results for {\sc Conn $k$-Horn}.
On the algorithmic side, we first present an exact exponential-time algorithm for arbitrary $k$ without any structural restrictions.
Our algorithm builds on the deterministic PPZ algorithm proposed by Paturi, Pudl\'{a}k, and Zane. 
It runs in $O^*(2^{(1 - 1/2k)n})$ time and polynomial space, achieving an exponential improvement over the previously known algorithm for the Boolean connectivity problem of $k$-CNF formulas, shown by Makino, Tamaki, and Yamamoto.
We next give two polynomial-time algorithms for arbitrary $k$ under the following two restrictions: (i) each variable appears at most twice, and (ii) each clause has length exactly $k$ and each variable appears at most $k$ times.
On the hardness side, we prove that {\sc Conn $3$-Horn} remains $\mathsf{coNP}$-complete even when each variable appears exactly three times.}

\section{Introduction}
The Boolean connectivity problem asks whether the satisfying assignments of a given Boolean formula are connected over the $n$-dimensional hypercube.
The structure and connectivity of satisfying assignments are closely related to the analysis of satisfiability algorithms and satisfiability thresholds. 
The satisfiability problem is characterized by the famous Schaefer dichotomy theorem~\cite{Schaefer78}.
Schaefer's class denotes the class of formulas, e.g., $2$-CNF, (dual) Horn, and affine formulas.
Ekin, Hammer, and Kogan~\cite{Ekin1999Boolean} showed that the Boolean connectivity of DNF is solvable in time linear in the size of DNF, while determining the connectivity of falsifying assignments of DNF is $\mathsf{coNP}$-hard, i.e., the Boolean connectivity of CNF is $\mathsf{coNP}$-hard.
Gopalan, Kolaitis, Maneva, and Papadimitriou~\cite{GopalanKMP09} extensively studied the computational complexity of the Boolean connectivity problem.
The paper~\cite{GopalanKMP09} showed that it is either $\mathsf{coNP}$-complete or $\mathsf{PSPACE}$-complete for non-Schaefer's class, and in $\mathsf{coNP}$ for Schaefer's class.
It has also been conjectured that the problem admits a polynomial-time algorithm for Schaefer’s class.
Makino, Tamaki, and Yamamoto~\cite{MakinoTY10} and Schwerdtfeger~\cite{Schwerdtfeger14} gave a negative answer by proving the $\mathsf{coNP}$-completeness of the Boolean connectivity problem for $k$-Horn and other related formulas.

Although the trichotomy theorem of the Boolean connectivity problem has been fully established, only a few algorithms are known for instances beyond $\mathsf{P}$.
For the connectivity of $k$-CNF formulas ({\sc Conn $k$-CNF}), it is solvable in polynomial time when $k=2$~\cite{GopalanKMP09}.
Makino et al.~\cite{MakinoTY11} presented a moderately exponential-time and exponential-apce algorithm for {\sc Conn $k$-CNF}~($k\ge3$) over $n$ variables and $m$ clauses which is $\mathsf{PSPACE}$-complete.
Their algorithm constructs a solution graph over partial satisfying assignments and subsequently determines its connectivity.
It runs in $O^*(2^{(1-c_k)n})$ time and exponential space, where $c_k = \log{\beta_k}/(1+\log{\beta_k})$ and $\beta_k~(< 2)$ is the largest positive real number that satisfies $x^k - x^{k-1} - \cdots - x - 1 = 0$. 
The $O^*$ notation suppresses polynomial factors in $n$ and $m$.
Throughout this paper, the base of the logarithm is 2.
Note that $c_k = 1/2^{O(k)}$.

In this paper, we deeply investigate the complexity of the Boolean connectivity of $k$-Horn formulas ({\sc Conn $k$-Horn}). 
We first consider an exact algorithm for {\sc Conn $k$-Horn} without any structural restrictions.
The algorithm of Makino et al. can also solve {\sc Conn $k$-Horn} in the same running time since $k$-Horn formulas are subsets of $k$-CNF formulas.
Consequently, it is natural to consider whether a faster algorithm for {\sc Conn $k$-Horn} exists.
Makino et al.~\cite{MakinoTY10} proposed an exact algorithm that runs in polynomial time with respect to the number of variables and the size of the characteristic set of a Horn formula.
Unfortunately, the efficiency of computing the characteristic set of a given Horn formula is not yet well understood.
We give an affirmative answer by utilizing the deterministic PPZ algorithm for $k$-SAT shown by Paturi, Pudl\'{a}k, and Zane~\cite{PaturiPZ99}.
The deterministic PPZ algorithm can find not only a satisfiability assignment but also all locally minimal satisfying assignments.
Moreover, {\sc Conn $k$-Horn} can be solved to find a locally minimal satisfying assignment with Hamming weight at least one~\cite{GopalanKMP09}.
These observations enable us to solve {\sc Conn $k$-Horn} exponentially faster than algorithms for {\sc Conn $k$-CNF}. Moreover, the proposed algorithm requires only polynomial space.
\begin{restatable}{theorem}{ppzalg}\label{thm:main2}
Given a $k$-Horn formula over $n$ variables, there exists a deterministic $O^*(2^{(1-1/2k)n})$ time and polynomial-space algorithm for {\sc Conn $k$-Horn}.
\end{restatable}

We next consider {\sc Conn $k$-Horn} with a bounded number of occurrences of each variable.
We present the following algorithmic and hardness results.
\begin{itemize}
\item A polynomial-time algorithm for {\sc Conn $k$-Horn} when each variable appears at most twice ({\sc Conn $k$-Horn-2)}.
The algorithm is based on the relationship between variables that appear only as negative literals and locally minimal satisfying assignments.
Any variable that appears only as negative literals is assigned the value $0$ in any locally minimal satisfying assignment.
We prove that determining the connectivity of the solution graph for the remaining formula after assigning 0 to such variables is solvable in polynomial time.
\item A polynomial-time algorithm for {\sc Conn $k$-Horn} when each clause have exactly $k$ literals and each variable appears at most $k$ times ({\sc Conn E$k$-Horn-$k$)}.
The algorithm is based on the relationship between the solution graph of a Horn formula and a non-empty maximal self-implicating set, as shown in~\cite{Schwerdtfeger14}.
Any variable in a self-implicating set is forced to be true when all the other variables in the set are true.
The solution graph is disconnected if and only if there exists a non-empty maximal self-implicating set and no clauses with only negative literals of these variables.
We show that such a set can be found in polynomial time.
\item The $\mathsf{coNP}$-completeness of {\sc Conn 3-Horn} even when each variable appears exactly three times ({\sc Conn $3$-Horn-E$3$}).
We first prove the complement of {\sc Conn $3$-Horn} remains $\mathsf{NP}$-complete even if each variable appears at most three times ({\sc Conn $3$-Horn-3}) by a polynomial-time reduction from {\sc Monotone Not-All-Equal 3-SAT}, where each clause has exactly three literals and each variable appears exactly four times ({\sc Monotone NAE E3-SAT-E4}). This problem is known to be $\mathsf{NP}$-complete~\cite{Darmann20}.
We then reduce the {\sc Conn $3$-Horn-3} to the {\sc Conn $3$-Horn-E3}.
\end{itemize}

\paragraph*{Related Work}
The connectivity problem between two satisfying assignments (\textsc{st-Conn}) has also been extensively studied.
Gopalan et al.~\cite{GopalanKMP09} showed that \textsc{st-Conn} is solvable in polynomial time for Schaefer's and a part of non-Schaefer's class and is $\mathsf{PSPACE}$-complete otherwise.
Scharpfenecker~\cite{Scharpfenecker2015Structure} proved that the complexity of SAT is equivalent to that of \textsc{st-Conn} in Schaefer's class.
This implies that \textsc{st-Conn} is $\mathsf{P}$-complete for Horn formulas and is $\mathsf{NL}$-complete for 2-CNFs.
Cardinal, Demaine, Eppstein, Hearn, and Winslow~\cite{Cardinal2018NAE} showed that \textsc{st-Conn} of \textsc{Planar Monotone Not-All-Equal 3-SAT} is $\mathsf{PSPACE}$-complete.  

The problem of finding the shortest path in the solution graph of Boolean formulas has also been investigated.
Mouawad, Nishimura, Pathak, and Raman~\cite{Mouawad15} studied the computational complexity of this problem and proved a trichotomy theorem showing that the problem lies in $\mathsf{P}$, is $\mathsf{NP}$-complete, or is $\mathsf{PSPACE}$-complete.
The paper~\cite{Mouawad15} also showed that there exist classes of Boolean formulas for which the shortest path can be found in polynomial time, even though its length is not equal to the symmetric difference between the values of the variables in $s$ and $t$.
Bonsma, Mouawad, Nishimura, and Raman~\cite{BonsmaMNR14} showed the problem is W[1]-hard when parameterized by the path length $\ell$ from $s$ to $t$.

The complexity of the isomorphism of two solution graphs was studied in~\cite{Scharpfenecker2016Isomorphism}.
This problem is $\mathsf{PSPACE}$-hard and lies in $\mathsf{EXP}$ for general formulas, and is $\mathsf{C_{=}P}$-complete for $2$-CNF. The class $\mathsf{C_{=}P}$ is a subclass of $\mathsf{PSPACE}$ and is defined in terms of exact counting.

\section{Preliminaries}
Let $x$ be a Boolean variable that takes true (1) or false (0).
A \emph{literal} is a Boolean variable $x$ (positive literal) or its negation $\bar{x}$ (negative literal).
A \emph{clause} is a disjunction of literals.
The \emph{length} of clause $C$ is defined as the number of literals in $C$.
A \emph{unit clause} is a clause of length one. 
A conjunctive normal form (CNF) formula is a conjunction of clauses, and a $k$-CNF formula is a CNF formula if the length of each clause is at most $k$.
A CNF formula is \emph{monotone} if all clauses in the formula consist of only positive (or negative) literals.
In this paper, we use Monotone CNF as CNF formulas with all positive literals unless otherwise stated.
A \emph{Horn clause} is a clause including at most one positive literal.
A Horn formula is a CNF formula whose all clauses are Horn, and a $k$-Horn formula is a $k$-CNF formula and a Horn formula.
Let $X = \{x_1, x_2, ..., x_n\}$ be a set of Boolean variables. 
An \emph{assignment} of a Boolean formula $\varphi$ over $X$ is $\alpha=(\alpha_1, \alpha_2, \ldots, \alpha_n) \in\{0,1\}^n$ such that $x_i=\alpha_i$ for all~$i$.
A \emph{partial assignment} of a Boolean formula $\varphi$ over $X$ is $\alpha=(\alpha_1, \alpha_2, \ldots, \alpha_n) \in\{0,1, *\}^n$ such that $x_i=\alpha_i$ for all~$i$. 
A partial assignment with $\alpha_i = *$ means variable $x_i$ is not assigned a 0/1 value.
For a variable $x$ and an assignment $\alpha$, we denote by $\alpha(x)$ the value of $x$ under an assignment $\alpha$.
We can simplify a Boolean formula $\varphi$ by a partial assignment $\alpha$ in a natural way, and the resulting formula is denoted by $\varphi|_\alpha$.
An (partial) assignment $\alpha$ is a \emph{(partial)~satisfying assignment} of a Boolean formula $\varphi$ if $\varphi|_\alpha$ is true.
The satisfiability problem is to determine whether there exists a satisfying assignment of a given Boolean formula $\varphi$.

Many excellent exponential-time algorithms for the satisfiability problem of $k$-CNF formulas ($k$-CNF SAT) are known.
One of the famous algorithms is shown by Paturi, Pudl\'{a}k, and Zane~\cite{PaturiPZ99}.
\begin{theorem}[\cite{PaturiPZ99}]
\label{theo:paturi}
    There exists an algorithm that can solve {\sc $k$-CNF SAT} over $n$ variables in a deterministic $O^*(2^{(1-1/2k)n})$ time and polynomial space.
\end{theorem}

For an assignment $\alpha\in \{0,1\}^n$, the \emph{Hamming weight} of 
$\alpha$ is $|\{i: \alpha_i=1~(1\leq i\leq  n)\}|$.
For two assignments $\alpha, \alpha' \in \{0,1\}^n$, the \emph{Hamming distance} between $\alpha$ and $\alpha'$, denoted by $d(\alpha, \alpha')$, is $|\{i: \alpha_i \neq \alpha'_i~(1\leq i\leq  n)\}|$.
We define the \emph{solution graph} of a Boolean formula $\varphi$ as $G_\varphi \coloneqq (V_\varphi, E_\varphi)$, where $V_\varphi$ is a set of all satisfying assignments of $\varphi$, that is, $V_\varphi = \{\alpha: \varphi|_{\alpha} = 1, \alpha \in \{0,1\}^n\}$, 
and $E_\varphi = \{(\alpha,\alpha'): d(\alpha,\alpha')=1 \text{~and~} \alpha, \alpha' \in V_{\varphi}\}$.
\begin{definition}
The Boolean connectivity problem {\sc(\sc Conn)} is to ask whether the solution graph $G_\varphi$ of a given Boolean formula $\varphi$ over $n$ variables is connected.
\end{definition}
Makino, Tamaki, and Yamamoto~\cite{MakinoTY10} showed that {\sc Conn $3$-Horn} is $\mathsf{coNP}$-complete, thus the following lemma holds.
\begin{lemma}[\cite{MakinoTY10}]
{\sc Conn $k$-Horn} is $\mathsf{coNP}$-complete when $k\geq 3$.
\end{lemma}
Given two assignments $\alpha, \alpha' \in \{0,1\}^n$, we denote $\alpha \le \alpha'$ as the coordinate-wise partial order if $\alpha_i \le \alpha'_i$ for each $1 \le i \le n$. 
A satisfying assignment $\alpha$ for a given formula $\varphi$ is \emph{locally minimal} if $\alpha$ has no neighboring satisfying assignment $\alpha'$ with $d(\alpha, \alpha') = 1$ and $\alpha' \le \alpha$. 
For $v, v' \in V_\varphi$, 
a \emph{monotone path} from $v$ to $v'$ is a path in $G_\varphi$, $v \rightarrow u_1 \rightarrow \cdots \rightarrow u_r \rightarrow v'$ such that $v\leq u_1 \leq \cdots \leq u_r \leq v'$. 
Gopalan, Kolaitis, Maneva, and Papadimitriou~\cite{GopalanKMP09} showed a notable connection between locally minimal satisfying assignments and the connectivity of Horn formulas.
\begin{lemma}[\cite{GopalanKMP09}]
\label{lem:local_minimal_gopalan}
Given a Horn formula $\varphi$, every connected component of $G_\varphi$ has a unique locally minimal satisfying assignment. Moreover, there exists a monotone path from the locally minimal satisfying assignment to each satisfying assignment in the same connected component.
\end{lemma}

The following corollary follows from that any Horn formula without unit clauses over $n$ variables has an all-zero satisfying assignment ($0^n$) and Lemma~\ref{lem:local_minimal_gopalan}.
\begin{corollary}[\cite{MakinoTY10}]
\label{col:local_minimal_makino}
Given a Horn formula $\varphi$ without unit clauses, the solution graph $G_\varphi$ is connected if and only if no locally minimal non-zero satisfying assignment exists.
\end{corollary}

Let $P(C)$ (resp. $N(C)$) denote the set of variables that appear in positive (resp. negative) literals in a clause $C$ of $\varphi$.
Makino et al.~\cite{MakinoTY10} introduced the following Boolean formula $\Phi_\varphi$.
\begin{align}\label{formula:makino}
\Phi_\varphi &= \varphi \wedge \bigwedge^n_{i=1} D_i,\\
D_i &= \overline{x_i} \vee \bigvee_{C \in \varphi : P(C) = \{x_i \} } \bigwedge_{y \in N(C)} y.\nonumber
\end{align}
For example, given a Horn formula
\[ \varphi=(x_1\vee \overline{x_2} \vee \overline{x_3})(x_1 \vee \overline{x_3}\vee \overline{x_4})(\overline{x_2} \vee \overline{x_3}\vee \overline{x_4})(\overline{x_2}\vee x_3\vee \overline{x_4})(\overline{x_1}\vee \overline{x_3}\vee x_4),\]
then
\[
D_1 = \overline{x_1}\vee x_2x_3 \vee x_3x_4, ~D_2 = \overline{x_2},~D_3 = \overline{x_3}\vee x_2x_4,~D_4 = \overline{x_4}\vee x_1x_3.
\]
Note that each $D_i$ is either a unit clause with a negative literal or a conjunction of one negative literal and disjunctions of positive literals.
The satisfiability of $\Phi_\varphi$ is strongly related to the connectivity of the solution graph of a Horn formula $\varphi$
because a satisfying assignment of $\Phi_\varphi$ is a locally minimal satisfying assignment of $\varphi$.
\begin{lemma}[\cite{MakinoTY10}]
\label{lem:makino_formula}
        Given a Horn formula $\varphi$ without unit clauses over $n$ variables, the solution graph $G_\varphi$ is disconnected if and only if $\Phi_\varphi$ has a non-zero satisfying assignment.
\end{lemma}

A \emph{restraint clause} is a clause with only negative literals, and a \emph{restraint set} is a set of variables included in restraint clauses.
An \emph{implication clause} is a clause that has one positive literal and one or more negative literals.
A Boolean variable $x$ is said to be \emph{implied} by a set of variables $U$ if setting all variables in $U$ to true forces $x$ to be true.
A \emph{self-implicating set} $U$ is a set of variables such that every $x \in U$ is implied by $U\setminus\{x\}$. 
We also say that $U$ is a \emph{maximal self-implicating set} if $U$ is a set of all variables implied by $U$. 
For maximal self-implicating sets, there exists a relation between locally minimal satisfying assignments and maximal self-implicating sets.

\begin{lemma}[\cite{Schwerdtfeger14}]\label{lem:sch2}
    For every Horn formula $\varphi$ without positive unit clauses, there is a bijection correlating each connected component $\varphi_i$ with a maximal self-implicating set $U_i$ containing no restraint set, i.e., $U_i$ consists of the variables assigned true in the minimum assignment of $\varphi_i$ (the 'lowest' component is correlated with the empty set).
\end{lemma}

The following corollary follows from Lemma~\ref{lem:sch2} and Corollary~\ref{col:local_minimal_makino}.

\begin{corollary}[\cite{Schwerdtfeger14}]\label{col:schwerd}
    The solution graph of a Horn formula $\varphi$ without positive unit clauses is disconnected if and only if $\varphi$ has a non-empty maximal self-implicating set containing no restraint set.
\end{corollary}

\section{PPZ-based Algorithm for {\sc\bfseries Conn $k$-Horn}}

We present an exact algorithm for {\sc Conn $k$-Horn} based on the deterministic PPZ algorithm~\cite{PaturiPZ99}.
Before presenting our new algorithm, we describe the notations and behaviors of the deterministic PPZ algorithm, referred to as {\sc DetPPZ}.
Let $\varphi$ be a $k$-CNF formula with $n$ variables and $m$ clauses, and let $\alpha$ be a satisfying assignment.
A clause $C_{(\alpha,i)}$ is a \emph{critical clause} for the variable $x_i$ at the solution $\alpha$ if $C_{(\alpha, i)}$ is to be false by flipping $\alpha_i$.
A variable $x_i$ is a \emph{critical variable} if there exists a critical clause $C_{(\alpha, i)}$.
The {\sc DetPPZ} first checks whether there exists a satisfying assignment of $\varphi$ to check all assignments with Hamming weight at most $\epsilon n$ where $0<\epsilon < 1/2$.
If there is no such satisfying assignment, the {\sc DetPPZ} tries to find a locally minimal satisfying assignment $\alpha$ with Hamming weight at least $\epsilon n$.
Note that $\varphi$ is falsified by any assignment obtained by flipping any 1 in $\alpha$ to 0.
Thus, $\alpha$ must have at least $\epsilon n$ critical variables.
Let $\sigma=(\sigma_1, \ldots, \sigma_n)$ be a permutation of $\{1,\ldots,n\}$.
The {\sc DetPPZ} picks up a permutation $\sigma$ in a `good' small sample space and branches by assigning a 0/1 value to $x_{\sigma_i}$ according to the order in $\sigma$.
If any critical variable appears last among the variables in the corresponding critical clause, then it can be assigned correctly.
It helps to reduce the number of branches in the algorithm. 
Modifying the {\sc DetPPZ} such that it does not halt even if it finds a satisfying assignment with Hamming weight at most $\epsilon n$, the {\sc DetPPZ} can find all locally minimal satisfying assignments. 
From Corollary~\ref{col:local_minimal_makino}, {\sc Conn $k$-Horn} can be solved by checking whether there exists some locally minimal satisfying assignment except for $0^n$ assignment. Thus, by combining the (modified) {\sc DetPPZ} and the procedure checking locally minimality, we solve {\sc Conn $k$-Horn} in the same running time as the {\sc DetPPZ}. 
We describe our algorithm, {\sc PPZ-Conn-$k$Horn}, in Algorithm~\ref{algo:conn_khorn2}.

\begin{algorithm}[t]
 \caption{{\sc PPZ-Conn-$k$Horn$(\varphi)$}}
 \label{algo:conn_khorn2}
 Construct the set of permutations $S$ in Lemma~\ref{lem:sample_space}
 
 \For{all assignments $\alpha$ with at most $\epsilon n$ 1's without all 0's}{
    \If{$\alpha$ satisfies $\varphi$ and $\alpha$ is locally minimal}{\Return NO}
 }
 \For{all permutations $\sigma$ in $S$}{
    \For{all strings $s$ of $n(1-\epsilon/k)+1$ bits}{
        \For{$i=1$ to $n$}{
        Let $x_j$ be the $i$-th variable according to $\sigma$~(i.e., $j=\sigma_i$).\;
        \If{there exists a unit clause $x_j$ or $\overline{x_j}$}{set $x_j$ to make that clause true}
        \ElseIf{there exists an unused bit from $s$}{set $x_j$ equal to the next unused bit from $s$}
        \Else{go to the next string in the for loop of lines 6--17. \\
        \Comment{there are no unused bits from $s$}}
        }
        \If{$\varphi$ is satisfiable and the satisfying assignment is locally minimal}{\Return NO}
    }
 }
\Return YES
\end{algorithm}

The following `good' sample space of permutations using $k$-wise independent variables is useful.
\begin{lemma}[\cite{PaturiPZ99}]\label{lem:sample_space}
One can construct the set of permutations $S$ of size $O(n^{3k})$ that satisfies the following condition: for any set $Y$ of up to $k$ variables, any variable $y\in Y$, and for a randomly chosen permutation from $S$, the probability that $y$ appears last among the variables in $Y$ is at least $1/|Y|-1/n$.
\end{lemma}
The following lemma immediately follows from the proof of correctness of the deterministic PPZ algorithm shown in \cite{PaturiPZ99}.
We provide almost the same proof for self-containedness in this paper.
\begin{lemma}\label{lem:completeness}
Let $S$ be the set of permutations constructed by Lemma~\ref{lem:sample_space}. 
Let $\alpha$ be an arbitrary satisfying assignment with at least $\epsilon n$ critical clauses.
There exists a permutation in $S$ that produces $\alpha$.
\end{lemma}
\begin{proof} 
Let $\alpha$ be a satisfying assignment with at least $\epsilon n$ critical clauses.
We now choose a permutation $\sigma$ from $S$ at random.
For each critical clause at $\alpha$, the probability that the critical variable $x$ appears last in $\sigma$ among the variables in the clause is at least $1/k-1/n$ by Lemma~\ref{lem:sample_space}. 
Since the number of critical clauses is at least $\epsilon n$, the expected number of times this event appears is at least $\epsilon n(1/k-1/n)=\epsilon n/k -\epsilon$, and there exists some $\sigma$ in $S$ that achieves at least the expectation. 
With respect to such $\sigma$, we assign at most $n-(\epsilon n/k -\epsilon)<n(1-\epsilon/k)+1$ variables to satisfy $\varphi$.  
Thus, by checking all strings of $n(1-\epsilon/k)+1$ bits, the permutation $\sigma$ can produce a satisfying assignment $\alpha$.
\end{proof}

We next show that the locally minimality of satisfying assignments can be efficiently checked. 
\begin{lemma}\label{lem:minimal_check}
Let $\alpha$ be a satisfying assignment of $\varphi$ with Hamming weight $w$.
One can check whether $\alpha$ is locally minimal in $O(mw+n)$ time.
\end{lemma}
\begin{proof}
From the definition of locally minimality, if any satisfying assignment $\alpha$ of $\varphi$ is locally minimal, then any assignment $\alpha'$ obtained by flipping any 1's bit is not a satisfying assignment of $\varphi$. 
It can be done by checking whether $\alpha'$ satisfies $\varphi$.
Since the Hamming weight of $\alpha$ is $w$, we can check the locally minimality of $\alpha$ in $O(mw+n)$ time.
\end{proof}

We present the main lemma, which establishes the correctness of our algorithm.
\begin{lemma}\label{lem:find_local_minimal}
The deterministic PPZ algorithm can find all locally minimal satisfying assignments of $\varphi$.
\end{lemma}
\begin{proof}
The deterministic PPZ algorithm first checks all assignments with Hamming weight at most $\epsilon n$. 
Hence, any satisfying assignment $\alpha$ with Hamming weight at most $\epsilon n$ can be found.
We can check whether $\alpha$ is locally minimal from Lemma~\ref{lem:minimal_check}.

We next consider that any locally minimal satisfying assignment $\alpha$ has Hamming weight at least $\epsilon n$.
Since $\alpha$ is a locally minimal satisfying assignment and has at least $\epsilon n$ 1's, there exist at least $\epsilon n$ critical clauses at $\alpha$.
From Lemma~\ref{lem:completeness}, $\alpha$ must be produced by some permutation in the sample space $S$ constructed in the algorithm.
The algorithm checks all permutations in $S$, then it can find $\alpha$.
We can check whether $\alpha$ is truly locally minimal from Lemma~\ref{lem:minimal_check}.
\end{proof}

We are now ready to prove Theorem~\ref{thm:main2}.
\ppzalg*
\begin{proof}
We can solve {\sc Conn $k$-Horn} by checking whether there exists some locally minimal non-zero satisfying assignment from Corollary~\ref{col:local_minimal_makino}.
Lemma~\ref{lem:find_local_minimal} shows that the deterministic PPZ algorithm can find all locally minimal satisfying assignments.
Thus, {\sc PPZ-Conn-$k$Horn} in Algorithm~\ref{algo:conn_khorn2} solves {\sc Conn $k$-Horn} correctly.
The deterministic PPZ algorithm runs in $O^*(2^{(1-1/2k)n})$ time and polynomial space. 
Checking locally minimality can be done $O(mn)$ time for each satisfying assignment from Lemma~\ref{lem:minimal_check}.
Thus, our algorithm runs in $O^*(2^{(1-1/2k)n})$ time and polynomial space.
\end{proof}

\begin{algorithm}[t]
 \caption{{\sc Count-CC-$k$Horn$(\varphi)$}}
 \label{algo:conn_khorn3}
 count $\leftarrow$ 0
 
 Construct the set of permutations $S$ in Lemma~\ref{lem:sample_space}
 
 \For{all assignments $\alpha$ with at most $\epsilon n$ 1's}{
    \If{$\alpha$ satisfies $\varphi$, and $\alpha$ is locally minimal and has never been produced}{count $\leftarrow$ count+1}
 }
 \For{all permutations $\sigma$ in $S$}{
    \For{all strings $s$ of $n(1-\epsilon/k)+1$ bits}{$\alpha \leftarrow 0^n$
    
        \For{$i=1$ to $n$}{
        Let $x_j$ be the $i$-th variable according to $\sigma$~(i.e., $j=\sigma_i$).
        
        \If{there exists a unit clause $x_j$ or $\overline{x_j}$}{set $\alpha_j$ to make that clause true}
        \ElseIf{there exists an unused bit from $s$}{set $\alpha_j$ equal to the next unused bit from $s$}
        \Else{go to the next string in the for loop of lines 6--17. \\
        \Comment{there are no unused bits from $s$}}
        }
        \If{$\varphi$ is satisfiable, and the satisfying assignment $\alpha$ is locally minimal and has never been produced}{count $\leftarrow$ count+1}
    }
 }
 \Return count
\end{algorithm}

By adding the procedure of counting the number of distinct locally minimal satisfying assignments (see Algorithm~\ref{algo:conn_khorn3}), we obtain the following Corollary.
\begin{corollary}\label{cor:main2}
    Counting the number of connected components in the solution graph of a given $k$-Horn formula over $n$ variables can be done in deterministic $O^*(2^{(1-1/2k)n})$ time and $O^*(2^{(1-1/2k)n})$ space.
\end{corollary}

\section{Complexity on the Boolean Connectivity of $k$-Horn Formulas with Bounded Variable Occurrence}

We investigate the complexity of {\sc Conn $k$-Horn} with bounded variable occurrences.
In Section~\ref{sec:poly_horn_2}, we present a polynomial-time algorithm for {\sc Conn $k$-Horn-2}, where each variable appears at most twice.
In Section~\ref{sec:poly_E3}, we provide a polynomial-time algorithm for {\sc Conn E$k$-Horn-3}, where each clause has exactly $k$ literals and each variable appears at most $k$ times.
In Section~\ref{sec:hardness}, we show the $\mathsf{coNP}$-completeness of {\sc Conn 3-Horn-E3}, where each variable appears exactly three times.

\subsection{Polynomial-time Algorithm for {\sc\bfseries Conn $k$-Horn-2}}\label{sec:poly_horn_2}

We denote $\varphi$ as a $k$-Horn formula over $n$ variables.
If the solution graph $G_{\varphi}$ of $\varphi$ is connected, then there exists a unique locally minimal satisfying assignment $0^n$ from Corollary~\ref{col:local_minimal_makino}.
Therefore, if $\varphi$ has a variable $x_i$ appearing only as a negative literal, we can show that the solution graph $G_{\varphi}$ is connected if and only if the solution graph $G_{\varphi|_{x_i=0}}$ is connected, where $\varphi|_{x_i=0}$ denotes $\varphi$ after assigning 0 to $x_i$.
This enables us to obtain a $k$-Horn formula in which each variable appears at least once as a positive literal.
Therefore, in the following, we can consider an instance for {\sc Conn $k$-Horn-2} as a $k$-Horn formula over $n$ variables where each variable appears at most twice and each variable appears at least once as a positive literal.
We also show that if there is a variable $x_i$ appearing in the clause $(x_i\vee\overline{x_i})$ in $\varphi$, then it does not appear in any other clause since each variable of $\varphi$ appears at most twice.
Therefore, we can eliminate the variable $x_i$ from $\varphi$ since the clause $x_i\vee\overline{x_i}$ is satisfied regardless of the value of $x_i$.
Throughout this section, we assume that any instance of {\sc Conn $k$-Horn-2} contains no such variables.
We now prove two auxiliary lemmas for the polynomial-time algorithm for {\sc Conn $k$-Horn-2}.

\begin{lemma}
\label{lem:horn-2_property}
Let $\varphi$ be a $k$-Horn formula over $n$ variables without unit clauses such that each variable appears at most twice in total, and at least once as a positive literal.
Then, $\varphi$ is a 2-Horn formula.
\end{lemma}

\begin{proof}
Since each variable appears at most twice in $\varphi$, the total number of literals in $\varphi$ is at most $2n$.
Also, $\varphi$ has at least $n$ implication clauses since each variable appears at least once as a positive literal.
We can show that the total number of literals in $\varphi$ is at least $2n$ since $\varphi$ has no unit clauses.
Therefore, the total number of literals in $\varphi$ is $2n$.
This means that the length of each clause in $\varphi$ is two since $\varphi$ has no unit clause.
This completes the proof.
\end{proof}

\begin{lemma}
\label{lem:horn-2_disconnected}
Let $\varphi$ be a $k$-Horn formula over $n$ variables without unit clauses such that each variable appears at most twice in total, and at least once as a positive literal.
Then, the solution graph $G_\varphi$ is disconnected.
\end{lemma}
\begin{proof}
From the lemma~\ref{lem:horn-2_property}, $\varphi$ is a 2-Horn formula.
Moreover, from the proof of Lemma~\ref{lem:horn-2_property}, the total number of literals in $\varphi$ is $2n$ and $\varphi$ has at least $n$ implication clauses.
Therefore, we can show that $\varphi$ has exactly $n$ implication clauses and there exists no restraint clauses.
This means that each clause has exactly one positive literal and exactly one negative literal.
Hence, the assignment $\alpha'$ assigned the value 1 to all variables in $\varphi$ is a satisfying assignment.
It also holds that any assignment obtained by flipping any 1's bit of $\alpha'$ is not a satisfying assignment of $\varphi$.
This concludes that $\alpha'$ is a locally minimal satisfying assignment.
This completes the proof since Corollary~\ref{col:local_minimal_makino} holds. 
\end{proof}

We present the polynomial-time algorithm based on Lemma~\ref{lem:horn-2_disconnected}.

\begin{theorem}
\label{thm:horn-2_poly}
{\sc Conn $k$-Horn-2} can be solved in polynomial time when each variable appears at most twice.
\end{theorem}

\begin{proof}
If a given $k$-Horn formula contains some unit clauses, then we repeatedly apply unit propagation until no unit clauses remain.
Therefore, we assume that a given Horn formula contains no unit clauses.
For a given $k$-Horn formula $\varphi$ over $n$ variables, if there exist variables that appear only as negative literals, we set these variables to 0.
We repeat this procedure until there are no such variables in $\varphi$.
Then, if all variables are assigned the value 0 in this procedure, this means that the solution graph $G_\varphi$ is connected from Corollary~\ref{col:local_minimal_makino} since $\varphi$ has no locally minimal non-zero satisfying assignments.
Otherwise, each variable in the remaining Horn formula appears at most twice in total, and at least once as a positive literal.
This means that the solution graph of the remaining formula is disconnected from Lemma~\ref{lem:horn-2_disconnected}.

We now estimate the running time of this algorithm.
Setting variables appearing only as negative literals to 0 can be done in $O(n)$ time since the number of clauses in $\varphi$ is $O(n)$.
We repeat this procedure at most $n$ times.
Hence, our algorithm runs in time $O(n^2)$.
\end{proof}

\subsection{Polynomial-time Algorithm for 
{\sc\bfseries Conn E$k$-Horn-$k$}}\label{sec:poly_E3}

For clarity of exposition, we consider {\sc Conn E$3$-Horn-$3$}, where each clause has exactly three literals and each variable appears at most three times.
For a set of variables $X$ and a CNF formula $\varphi$, we denote by $C_{\varphi[X]}$ the set of all clauses in $\varphi$ that consist of only variables in $X$.  

\begin{lemma}\label{lem:polyk3}
Let $\varphi$ be an instance of {\sc Conn E3-Horn-3} over $X=\{x_1,x_2,\dots,x_n\}$. The solution graph $G_{\varphi}$ is disconnected if and only if there exists a partition of the set $X$ into two sets $U$ and $X\setminus U$ and a partition of the set $C_{\varphi[X]}$ of clauses of $\varphi$ into two sets $C_{\varphi[U]}$ and $C_{\varphi[X\setminus U]}$ such that $U$ is a non-empty set of variables which appear once as a positive literal. 
\end{lemma}

\begin{proof}
We first prove the if-part, 
and assume that a variable set $U$ satisfies the if-condition.
We show that $U$ is a non-empty maximal self-implicating set containing no restraint set; then $G_\varphi$ is disconnected from Corollary~\ref{col:schwerd}.
The number of implication clauses in $C_{\varphi[U]}$ is $|U|$ since each variable in $U$ appears once as a positive literal.
These implication clauses have $3|U|$ literals since the length of each clause in $C_{\varphi[U]}$ is three, and then the total number of literals of clauses in $C_{\varphi[U]}$ is at least $3|U|$.
However, the total number of occurrences of variables in $U$ is at most $3|U|$ since each variable appears at most three times.
These lead to the total number of literals in $C_{\varphi[U]}$ being exactly $3|U|$.
Thus, all literals in $C_{\varphi[U]}$ appear in implication clauses in $C_{\varphi[U]}$, i.e., there is no restraint set in $U$.
For any implication clause $C$, the positive literal $x$ in $C$ is assigned true when the other negative literal(s) are assigned false.
This implies that every $x\in U$ is implied by $U\setminus \{x\}$.
In addition, any variable in any clause in $C_{\varphi[X\setminus U]}$ cannot be implied by $U$ since $U \cap (X\setminus U) = \emptyset$. 
Hence, $U$ is a non-empty maximal self-implicating set and contains no restraint set. 
This concludes the proof of the if-part.
    
We next prove the only-if part.
Assume that $G_\varphi$ is disconnected.
From Corollary~\ref{col:schwerd}, $\varphi$ has a non-empty maximal self-implicating set $U$ containing no restraint set. 
This implies that every variable in $U$ appears as a positive literal at least once.
We claim that the number of implication clauses with only variables in $U$ is $|U|$.
There exist at least $|U|$ implication clauses in $C_{\varphi[X]}$ since $U$ is a self-implicating set, i.e., there exists a clause $u \vee \overline{x}\vee \overline{y}$ for every $u\in U$.
Since the total number of occurrences of variables in $U$ is at most $3|U|$ and the length of each clause is three, there exist at most $|U|$ implication clauses with only variables in $U$.
If the number of such implication clauses with only variables in $U$ is at most $|U|-1$, there exists some $u\in U$ appearing in an implication clause as positive literals with at least one variable $x$ in $X\setminus U$.
We denote $(u\vee \overline{x} \vee \overline{y})$ as such a clause.
In this case, $u$ is implied by $x$ and $y$.
This contradicts the fact that $U$ is a maximal self-implicating set.
Thus, the number of implication clauses with only variables in $U$ is $|U|$, and this implies that every variable in $U$ appears as a positive literal only once.
The number of literals of these clauses is $3|U|$, and it equals the maximum total number of occurrences of variables in $U$.
This implies that the pair ($U$, $X\setminus U$) is a partition of $X$, and the pair of $(C_{\varphi[U]}$, $C_{\varphi[X\setminus U]})$ is a partition of $C_{\varphi[X]}$.
This concludes the proof of the only-if part.
\end{proof}

Now, we present a polynomial-time algorithm based on Lemma~\ref{lem:polyk3}.

\begin{theorem}\label{thm:main1}
{\sc Conn $3$-Horn} is solvable in polynomial time when each clause has length exactly three and each variable appears at most three times.
\end{theorem}
\begin{proof}
A given 3-Horn formula $\varphi$ over $X=\{x_1, x_2, \ldots, x_n\}$ has at most $3n$ literals and $n$ clauses since the length of each clause is three and each variable appears at most three times.
We find a non-empty variable set that satisfies Lemma~\ref{lem:polyk3} to determine the connectivity of the solution graph of $\varphi$.

We first enumerate the non-empty variable set $U$ such that $C_{\varphi[X]} = C_{\varphi[U]} \cup C_{\varphi[X\setminus U]}$ and $C_{\varphi[U]} \cap C_{\varphi[X\setminus U]}=\emptyset$ as follows.
We construct the graph $G=(X, E)$, where $E=\{(x_i, x_j) : \text{$C_{\varphi[X]}$ has a clause containing $x_i, x_j\in X$}\}$.
It is easy to see that the variable set $U'$ corresponding to the vertices in any connected component of $G$ constructs $C_{\varphi[U']}$.
Thus, we can enumerate the non-empty variable sets to enumerate the connected components of $G$.
We use the depth-first search for $G$ to enumerate all the connected components.
It remains to check whether each $U'$ is a set of variables that appear once as a positive literal.
If at least one such $U'$ exists, the solution graph is disconnected.
Otherwise, the solution graph is connected, as shown in Lemma~\ref{lem:polyk3}.

We now estimate the running time of this algorithm.
Constructing the graph $G$ can be done in $O(n)$ time since the number of edges is at most $3n$.
Enumerating the non-empty variable set $U'$ can be done in time $O(n)$ by the depth-first search.
Note that the number of non-empty variable sets is $O(n)$.
We can check whether each $U'$ is a set of variables that appear once as a positive literal in $O(n)$ by checking the clauses.
Hence, our algorithm runs in time $O(n^2)$.
\end{proof}

We can show Corollary~\ref{col:polyk} since the discussion in Lemma~\ref{lem:polyk3} also holds for $k \ge 3$.

\begin{corollary}\label{col:polyk}
Let $\varphi$ be an instance of {\sc Conn E$k$-Horn-$k$} over $X=\{x_1,x_2,\dots,x_n\}$. The solution graph $G_{\varphi}$ is disconnected if and only if there exists a partition of the set $X$ into two sets $U$ and $X\setminus U$ and a partition of the set $C_{\varphi[X]}$ of clauses of $\varphi$ into two sets $C_{\varphi[U]}$ and $C_{\varphi[X\setminus U]}$ such that $U$ is a non-empty set of variables which appear once as a positive literal.
\end{corollary}

Gopalan et al.~\cite{GopalanKMP09} showed that {\sc Conn 2-CNF} is solvable in polynomial time with no bounded variable occurrence. The proof of Theorem~\ref{thm:main1} also holds for $k \ge 3$ from Corollary~\ref{col:polyk}. Therefore, we can show Corollary~\ref{col:kpoly_algo}.

\begin{corollary}\label{col:kpoly_algo}
    For $k \ge 2$, {\sc Conn E$k$-Horn-$k$} is solvable in polynomial time.
\end{corollary}

We can show the following corollary from Corollary~\ref{col:polyk}. 
\begin{corollary}\label{col:upperbound_components}
    For $k \ge 3$, given a $k$-Horn formula $\varphi$, where each clause length and each variable occurrence are exactly $k$, the number of connected components in the solution graph of $\varphi$ is at most $2^{\lfloor n/k \rfloor}$.
\end{corollary}
\begin{proof}
Let $\varphi$ be an instance of {\sc Conn $k$-Horn} with each clause length exactly $k$ and each variable occurrence exactly $k$, where the solution graph of $\varphi$ is disconnected.
A variable set of $\varphi$ is denoted by $X$.
Then, there exists at least one variable set $U\subseteq X$ such that $C_{\varphi[U]}$ and $C_{\varphi[X\setminus U]}$ partition $C_{\varphi[X]}$ and $U$ is a non-empty set of variables that appear once as a positive literal.
Now, we assume that $\varphi$ has $m$ variable sets $U_i\ (1\le i\le m)$ and a variable set $V$ such that $C_{\varphi[X]} = \bigcup_{i=1}^m C_{\varphi[U_i]} \cup C_{\varphi[V]}$ and $\bigcup_{i=1}^m U_i\cup V=X, U_i\cap V=\emptyset$, and $U_i\cap U_j=\emptyset$ for any $i,j~(i\neq j)$, where $U_i$ is a non-empty set of variables which appear once as a positive literal.
Then, we can construct all locally minimal satisfying assignments by setting the value 0 or 1 to all variables for each $U_i\ (1\le i\le m)$ and setting the value 0 to all variables in $V$ from Lemma~\ref{lem:sch2}. Therefore, there exist $2^m$ locally minimal satisfying assignments of $\varphi$.
Moreover, $m$ is a multiple of $k$ since the length of each clause and the number of occurrences of each variable are exactly $k$.
Hence, the maximum value of $m$ is $\lfloor n/k \rfloor$.
\end{proof}

\subsection{coNP-completeness of {\sc\bfseries Conn 3-Horn-E3}}\label{sec:hardness}

We first prove the $\mathsf{coNP}$-completeness of {\sc Conn 3-Horn-3}, where each variable appears at most three times, and we then modify the instance so that each variable appears exactly three times.  
Instead of proving $\mathsf{coNP}$-completeness of {\sc Conn 3-Horn-3}, we show $\mathsf{NP}$-completeness of the complement problem of {\sc Conn 3-Horn-3}~({\sc DisConn 3-Horn-3}) that asks whether the solution graph of a given 3-Horn formula with each variable appearing at most three times is disconnected.
The proof is based on a polynomial-time reduction from {\sc Monotone NAE E3-SAT-E4}, which is $\mathsf{NP}$-complete, shown in~\cite{Darmann20}.
An \emph{NAE-satisfying assignment} is a satisfying assignment in which each clause has at least one literal assigned the value 0 and at least one literal assigned the value 1.
Given a Monotone 3-CNF formula $\phi$ with $n$ variables and $m$ clauses in which each clause has exactly three literals and each variable appears exactly four times, {\sc Monotone NAE E3-SAT-E4} asks whether there exists an NAE-satisfying assignment of $\phi$.
Our reduction utilizes the framework of Schwerdtfeger~\cite{Schwerdtfeger14}.

We give the construction of a 3-Horn formula from an instance of {\sc Monotone NAE E3-SAT-E4}.
Then, the 3-Horn formula satisfies that each clause in the instance of {\sc Monotone NAE E3-SAT-E4} has at least one literal assigned the value 0 and at least one literal assigned the value 1.
Also, each variable in the 3-Horn formula appears at most three times.
Let $X=\{x_1, x_2, \dots, x_n\}$ be a set of variables.
Let $\phi$ be an instance of {\sc Monotone NAE E3-SAT-E4} over $X$ with $m$ clauses. 
Let $C_\phi = \{C_1, C_2, \dots, C_m \}$ be the set of clauses of the formula $\phi$.
For simplicity, if the variable $x_{j}$ appears in the clause $C_i$, we write $x_{j}\in C_i$.
We first replace the $k$-th occurrence of $x_j \in X$ with a new variable $x_{j,k}$, and denote by $X_\phi$ the set of new variables created from $X$.
We construct a 3-Horn formula $\varphi_\phi=\varphi_1 \wedge \varphi_2 \wedge \varphi_3 \wedge \varphi_4$.

The formula $\varphi_1$ over $X_\phi$ is as follows:
\begin{equation*}
    \varphi_1=\bigwedge_{i=1}^{m}\bigvee_{x_{j,k}\in C_i} \overline{x_{j,k}}.
\end{equation*}
It ensures that each clause in $\phi$ has at least one literal assigned the value 0.

The formula $\varphi_2$ is a conjunction of $\varphi_{2,i}$ corresponding to each $C_i\in C_\phi$ which contains $x_{j,a}, x_{k,b}, x_{\ell, c}$.
It ensures that there exists a bijection between the set of NAE-satisfying assignments of $\phi$ and the set of locally minimal non-zero satisfying assignments of $\varphi_\phi$.
Let $X_\phi' = \{x_{i,j}' : 1\le i\le n, 1\le j\le 4 \}$, $Y =\{y_{i,j} : 1\le i\le m, 1\le j\le 3 \}$, $Y' =\{y_{i,j}' :  1\le i\le m, 1\le j\le 3 \}$, $P=\{p_{i,j} : 1\le i\le m, 1\le j\le n \}$, and $Q=\{q_{i,j} : 1\le i\le m, 1\le j\le n \}$.
The formula $\varphi_2$ over $X_\phi \cup X'_{\phi} \cup Y \cup Y' \cup P \cup Q$ is as follows:

\begin{align*}
    \varphi_2 = \bigwedge_{i=1}^{m}\varphi_{2,i} =& \bigwedge_{i=1}^{m}((\overline{y_{i,1}} \vee p_{i,j})(\overline{p_{i,j}} \vee \overline{x_{j,a}} \vee q_{i,j})(\overline{q_{i,j}} \vee x_{j,a}')(\overline{q_{i,j}} \vee y_{(i\bmod m)+1,1}')
    \\ &\wedge (\overline{y_{i,2}} \vee p_{i,k})(\overline{p_{i,k}} \vee \overline{x_{k,b}} \vee q_{i,k})(\overline{q_{i,k}} \vee x_{k,b}')(\overline{q_{i,k}} \vee y_{(i\bmod m)+1, 2}')
    \\ &\wedge (\overline{y_{i,3}} \vee p_{i,\ell})(\overline{p_{i,\ell}} \vee \overline{x_{\ell, c}} \vee q_{i,\ell})(\overline{q_{i,\ell}} \vee x_{\ell, c}')(\overline{q_{i,\ell}} \vee y_{(i\bmod m)+1, 3}')).
\end{align*}

The formula $\varphi_3$ is a conjunction of $\varphi_{3,j}$ corresponding to each $x_j\in X$.
It ensures that $x_{j,1}=x_{j,2}=x_{j,3}=x_{j,4}$ for each $j\ (1\le j\le n)$.
Let $U=\{u_{j,k} : 1\le j\le n, 1\le k\le 4\}$ and $D_X = \{d_{j} : 1\le j\le n \}$. 
The formula $\varphi_3$ over $D_X \cup X_\phi \cup X'_{\phi} \cup U$ is as follows:

\begin{align*}
    \varphi_3 = \bigwedge_{j=1}^{n}\varphi_{3,j} =&\bigwedge_{j=1}^{n}((u_{j,1}\vee \overline{x_{j,1}'} \vee \overline{x_{j,2}'})(u_{j,2}\vee \overline{x_{j,3}'} \vee \overline{x_{j,4}'})(d_j\vee \overline{u_{j,1}} \vee \overline{u_{j,2}})
    \\ &\wedge (u_{j,3}\vee \overline{d_{j}})(u_{j,4}\vee \overline{d_{j}})(x_{j,1}\vee \overline{u_{j,3}})(x_{j,2}\vee \overline{u_{j,3}})(x_{j,3}\vee \overline{u_{j,4}})(x_{j,4}\vee \overline{u_{j,4}})).
\end{align*}

Recall that a self-implicating set $S$ is a set of variables such that every $x \in S$ is implied by $S \setminus \{x\}$.
The following four claims establish properties of a non-empty maximal self-implicating set for the formula above.
All proofs proceed by repeatedly invoking the argument that the existence of a clause in the formula guarantees that at least one variable appearing in the clause belongs to a non-empty maximal self-implicating set.

\begin{claim}
\label{clm:x_gadget_property}
Let $S$ be an arbitrary non-empty maximal self-implicating set.
Then, the following properties hold:
\begin{enumerate}[(a)]
    \item If $S$ contains at least one $x_{j,k}\ (1\le k\le 4)$, then $S$ contains all variables appearing in $\varphi_{3,j}$,
    \item If $S$ contains $d_j$, then $S$ contains all $x_{j,k}$ and at least one $x'_{j,k}$,
    \item If $S$ contains at least one $u_{j,k}$, then $S$ contains at least one $x'_{j,k}$.
\end{enumerate}
\end{claim}

\begin{proof}
We show each property.
\begin{enumerate}[(a)]
    \item Without loss of generality, we assume that $S$ contains $x_{j,1}$.
    The positive literal $x_{j,1}$ appears only in the clause $(x_{j,1}\vee\overline{u_{j,3}})$ of $\varphi_{3,j}$, then $S$ contains $u_{j,3}$ to imply $x_{j,1}$ by variables in $S\setminus\{x_{j,1}\}$.
    By a similar argument and the clauses $(d_j\vee\overline{u_{j,3}}), (x_{j,2}\vee\overline{u_{j,3}}), (d_{j}\vee\overline{u_{j,1}}\vee\overline{u_{j,2}}), (u_{j,1}\vee\overline{x'_{j,1}}\vee\overline{x'_{j,2}})$, and $(u_{j,2}\vee\overline{x'_{j,3}}\vee\overline{x'_{j,4}})$, $S$ contains the variables $d_j, x_{j,2}, u_{j,1}, u_{j,2}, x'_{j,1}, x'_{j,2}, x'_{j,3}$, and $x'_{j,4}$.
    From the maximality of $S$ and the clause $(u_{j,4}\vee\overline{d_j})$, $S$ contains $u_{j,4}$.
    It means that $S$ contain $x_{j,3}$ and $x_{j,4}$ since $u_{j,4}$ implies both two variables from the clauses $(x_{j,3}\vee\overline{u_{j,4}}),(x_{j,4}\vee\overline{u_{j,4}})$.
    Therefore, $S$ contain all variables appearing in $\varphi_{3,j}$.
    \item We first show that $S$ contains all $x_{j,k}$.
    If $S$ contains $d_j$, then the variables $u_{j,3}$ and $u_{j,4}$ are contained from the maximality of $S$ and the clauses $(u_{j,3}\vee \overline{d_{j}}),(u_{j,4}\vee \overline{d_{j}})$.
    By a similar argument and the clauses $(x_{j,1}\vee \overline{u_{j,3}}),(x_{j,2}\vee \overline{u_{j,3}}),(x_{j,3}\vee \overline{u_{j,4}}),(x_{j,4}\vee \overline{u_{j,4}})$, $S$ contains $x_{j,1}, x_{j,2}, x_{j,3}$ and $x_{j,4}$.
    We next show that $S$ contains at least one $x'_{j,k}$.
    By a similar argument of (a) and the clauses $(d_{j}\vee\overline{u_{j,1}}\vee\overline{u_{j,2}}), (u_{j,1}\vee\overline{x'_{j,1}}\vee\overline{x'_{j,2}})$, and $(u_{j,2}\vee\overline{x'_{j,3}}\vee\overline{x'_{j,4}})$, $S$ contains the variables $u_{j,1}, u_{j,2}, x'_{j,1}, x'_{j,2}, x'_{j,3}$, and $x'_{j,4}$ if $S$ contains $d_j$.
    Hence, $S$ contains at least one $x'_{j,k}$.
    \item We consider two cases: (i) $S$ contains $u_{j,1}$ or $u_{j,2}$ (ii) $S$ contains $u_{j,3}$ or $u_{j,4}$.
    In Case~(i), $S$ contains the variables $x'_{j,1}, x'_{j,2}$ or $x'_{j,3}, x'_{j,4}$ from the clauses $(u_{j,1}\vee \overline{x'_{j,1}} \vee \overline{x'_{j,2}}), (u_{j,2}\vee \overline{x'_{j,3}} \vee \overline{x'_{j,4}})$.
    In Case~(ii), $S$ contains $d_j$ from the clauses $(u_{j,3}\vee \overline{d_{j}}), (u_{j,4}\vee \overline{d_{j}})$.
    Therefore, $S$ contains at least one $x'_{j,k}$ from (b).
\end{enumerate}

\end{proof}

The formula $\varphi_4$ is a conjunction of $\varphi_{4,i}$ corresponding to each $C_i \in C_{\phi}$.
It ensures that each clause in $\phi$ has at least one literal assigned the value 1.
Let $V=\{v_{i,j} : 1\le i\le m, 1\le j\le 4\}$ and $E_Y=\{e_i : 1\le i\le m\}$.
The formula $\varphi_{4}$ over $E_Y \cup Y \cup Y' \cup V$ is as follows:

\begin{align*}
    \varphi_4 = \bigwedge_{i=1}^{m} \varphi_{4,i}
    = &\bigwedge_{i=1}^{m} ((v_{i,1}\vee \overline{y_{i,1}'})(v_{i,1}\vee \overline{y_{i,2}'})(v_{i,2}\vee \overline{y_{i,3}'})
    (v_{i,2}\vee \overline{v_{i,1}})(e_i \vee \overline{v_{i,2}})(v_{i,3}\vee \overline{e_i})
    \\ &\wedge (v_{i,4}\vee \overline{v_{i,3}})(y_{i,1}\vee \overline{v_{i,4}})(y_{i,2}\vee \overline{v_{i,4}})(y_{i,3}\vee \overline{v_{i,3}})).
\end{align*}
From the construction of $\varphi_{4,i}$, we can get the following observation.

\begin{claim}
\label{clm:y_gadget_property}
Let $S$ be an arbitrary non-empty maximal self-implicating set.
Then, the following properties hold:
\begin{enumerate}[(a)]
    \item If $S$ contains at least one $y_{i,k}$, then $S$ contains $e_i$ and at least one $y'_{i,k}\ (1\le k\le 3)$,
    \item If $S$ contains at least one $v_{i,\ell}\ (1\le \ell\le 4)$, then $S$ contains $e_i$,
    \item If $S$ contains at least one $y'_{i,k}$, then $S$ contains $e_i$,
    \item If $S$ contains $e_i$, then $S$ contains at least one $y'_{i,k}\ (1\le k\le 3)$.
\end{enumerate}
\end{claim}

\begin{proof}
We show each property.
\begin{enumerate}[(a)]
    \item Without loss of generality, we assume that $S$ contains $y_{i,1}$.
    The positive literal $y_{i,1}$ appears only in the clause $(y_{i,1}\vee\overline{v_{i,4}})$ of $\varphi_{4,i}$, then $S$ contains $v_{i,4}$ to imply $y_{i,1}$ by variables in $S\setminus\{y_{i,1}\}$.
    By a similar argument and the clauses $(v_{i,4}\vee\overline{v_{i,3}}), (v_{i,3}\vee\overline{e_i})$, and $(e_i\vee\overline{v_{i,2}})$, $S$ contains the variables $v_{i,4}, v_{i,3}, e_i$, and $v_{i,2}$.
    The variable $v_{i,2}$ appears in the clauses $(v_{i,2}\vee\overline{y'_{i,3}})$ and $(v_{i,2}\vee\overline{v_{i,1}})$, then $S$ contains $y'_{i,3}$ or $v_{i,1}$.
    If $S$ contains $y'_{i,3}$, then the proof is done.
    Otherwise, $S$ contains $v_{i,1}$. 
    In this case, $S$ contains $y'_{i,1}$ or $y'_{i,2}$ to imply $v_{i,1}$ by variables in $S\setminus\{v_{i,1}\}$.
    Therefore, $S$ contains at least one $y'_{i,k}\ (1\le k\le 3)$.
    \item We consider two cases: (i) $S$ contains $v_{i,1}$ or $v_{i,2}$ (ii) $S$ contains $v_{i,3}$ or $v_{i,4}$.
    \subparagraph{Case~(i)}
    If $S$ contains $v_{i,2}$, then $S$ contains $e_i$ from the maximality of $S$ and the clause $(e_i\vee\overline{v_{i,2}})$.
    By a similar argument and the clauses $(v_{i,2}\vee\overline{v_{i,1}}), (e_i\vee\overline{v_{i,2}})$, if $S$ contains $v_{i,1}$, then $v_{i,2}$ and $e_i$ are contained in $S$.
    \subparagraph{Case~(ii)}
    The set $S$ contains $e_j$ from the clause $(v_{i,3}\vee\overline{e_i})$ if $S$ contains $v_{i,3}$.
    If $S$ contains $v_{i,4}$, then $S$ contains $v_{i,3}$ from the clause $(v_{i,4}\vee\overline{v_{i,3}})$.
    This means that $S$ contains $e_j$ by a similar argument.
    \item Without loss of generality, we assume that $S$ contains $y'_{i,1}$.
    Then, from the maximality of $S$ and the clauses $(v_{i,1}\vee\overline{y'_{i,1}}), (v_{i,2}\vee\overline{v_{i,1}})$, $S$ contains $v_{i,1}$ and $v_{i,2}$.
    By a similar argument and the clause $(e_i\vee\overline{v_i,2})$, $S$ contains $e_i$.
    \item By a similar argument to (a) and the clause $(e_i\vee\overline{v_{i,2}})$, $S$ contains $v_{i,2}$.
    The variable $v_{i,2}$ appears in the clauses $(v_{i,2}\vee\overline{y'_{i,3}})$ and $(v_{i,2}\vee\overline{v_{i,1}})$, then $S$ contains $y'_{i,3}$ or $v_{i,1}$.
    If $S$ contains $y'_{i,3}$, then the proof is done.
    Otherwise, $S$ contains $v_{i,1}$. 
    In this case, $S$ contains $y'_{i,1}$ or $y'_{i,2}$ to imply $v_{i,1}$ by variables in $S\setminus\{v_{i,1}\}$.
    Therefore, $S$ contains at least one $y'_{i,k}\ (1\le k\le 3)$.
\end{enumerate}

\end{proof}

Clearly, $\varphi_\phi$ is a 3-Horn formula.
The number of occurrences of each $d_j \in D_X$, $x_{j,k}\in X_\phi$, $q_{i,j}\in Q$, and $v_{i,j} \in V$ is exactly three, and other variables appear at most three times.
Hence, $\varphi_\phi$ is an instance of {\sc DisConn 3-Horn-3}.
We show the following two claims about maximal self-implicating sets of $\varphi_\phi$.

\begin{claim}
\label{clm:instance_property1}
Let $S$ be an arbitrary non-empty maximal self-implicating set of $\varphi_\phi$.
Then, $S$ contains all $e_i \in E_Y$ and all $p_{i,j}\in P$.
\end{claim}

\begin{proof}
To prove this claim, we show that $S$ contains all $e_i$.
If $S$ contains all $e_i$, then $S$ contains all $v_{i,3}$ from the maximality of $S$ and the clause $(v_{i,3}\vee\overline{e_i})$ of $\varphi_{4,i}$.
The variables $v_{i,4}$ and $y_{i,3}$ are contained from the maximality of $S$ and the clauses $(v_{i,4}\vee\overline{v_{i,3}}),(y_{i,3}\vee\overline{v_{i,3}})$.
Then, $S$ contains $y_{i,1}$ and $y_{i,2}$ which are implied by $v_{i,4}$ in the clauses $(y_{i,1}\vee\overline{v_{i,4}}),(y_{i,2}\vee\overline{v_{i,4}})$.
Now, $S$ contains the variables $y_{i,1}, y_{i,2}$, and $y_{i,3}$.
From the clauses $(\overline{y_{i,1}}\vee p_{i,j}),(\overline{y_{i,2}}\vee p_{i,k}),(\overline{y_{i,3}}\vee p_{i,\ell})$ and the maximality of $S$, the variables $p_{i,j}\in P$ are contained in $S$.
We prove that $S$ contains all $e_i$ in two steps: (i) $S$ has at least one $e_i$ (ii) If $S$ has at least one $e_i$, then $S$ contains all $e_i$.

We first show (i).
For the sake of contradiction, we assume that there exists a non-empty maximal self-implicating set $S$ that contains no element $e_i\in E_Y$.
Then, we consider cases depending on which variable sets of $\varphi_\phi$ the non-empty maximal self-implicating set $S$ contains at least one variable from.
The formula $\varphi_\phi$ has the set of variables $X_\phi, X'_\phi, Y, Y', P, Q, U, V, D_X$, and $E_Y$.
The formula $\varphi_{4,i}$ has the set of variables $E_Y, Y, Y'$, and $V$.
The formula $\varphi_{3,j}$ has the set of variables $D_X, X_\phi, X'_\phi$, and $U$.
The remaining sets of variables in $\varphi_\phi$ are $P$ and $Q$.
Thus, we consider the following four cases:
\begin{enumerate}
    \item $S$ has at least one variable except for $e_{i}$, appearing in $\varphi_{4,i}$,
    \item $S$ has at least one variable appearing in $\varphi_{3,j}$,
    \item $S$ has at least one variable in $P$,
    \item $S$ has at least one variable in $Q$.
\end{enumerate}
\subparagraph{Case 1}
The formula $\varphi_{4,i}$ has the set of variables $E_Y, Y, Y'$, and $V$.
If $S$ contains at least one variable in $Y$, then $S$ contains some $e_i\in E_Y$ from the property~(a) in Claim~\ref{clm:y_gadget_property}.
If $S$ contains at least one variable in $Y'$, then $S$ contains some $e_i\in E_Y$ from the property~(a) in Claim~\ref{clm:y_gadget_property}.
Also, if $S$ contains at least one variable in $V$, then $S$ contains some $e_i\in E_Y$ from the property~(b) in Claim~\ref{clm:y_gadget_property}.
This contradicts the assumption.

\subparagraph{Case 2}
The formula $\varphi_{3,j}$ has the set of variables $D_X, X_\phi, X'_\phi$, and $U$.
If $S$ contains any variable appearing in $\varphi_{3,j}$, then $S$ contains at least one $x'_{j,k}\ (1\le k\le 4)$ from the properties~(a), (b), and (c) in Claim~\ref{clm:x_gadget_property}.
The positive literal $x'_{j,k}$ appears only in $(\overline{q_{i,j}}\vee x'_{j,k})$ of $\varphi_{2,i}$, then $S$ contains $q_{i,j}$ to imply $x'_{j,k}$ by variables in $S\setminus\{x'_{j,k}\}$.
Therefore, since positive literal $q_{i,j}$ appears only in $(\overline{p_{i,j}}\vee \overline{x_{j,k}}\vee q_{i,j})$ of $\varphi_{2,j}$, $S$ contains $p_{i,j}$ and $x_{j,k}$.
The clause $(\overline{y_{i, a}}\vee p_{i,j})$ of $\varphi_{2,i}$ ensures that $S$ contains $y_{i,a}$ to imply $p_{i,j}$ by variables in $S\setminus \{p_{i,j}\}$.
Then, $S$ contains $e_{i}$ from the property~(a) in Claim~\ref{clm:y_gadget_property}.
This contradicts the assumption.

\subparagraph{Case 3 and Case 4}
We can consider the latter two cases as subcases of Case~2.
Therefore, we can show both cases by using the same discussion as Case~2.\\

We next show (ii).
We assume that there exists a non-empty maximal self-implicating set $S$ which does not contain some $e_{i}\in E_Y$.
From the above discussion, $S$ contains at least one $e_{i}\in E_Y$.
Then, $S$ contains at least one $y'_{i, j}$ from the property (d) in Claim~\ref{clm:y_gadget_property}.
Let $\ell=((i+m-2)\bmod m)+1$ and note that $\ell=m$ if $i=1$ and $\ell=i-1$ otherwise. 
From the clauses $(\overline{q_{\ell, j}} \vee y'_{i ,1}),(\overline{q_{\ell, k}} \vee y'_{i, 2}),(\overline{q_{\ell,a}} \vee y'_{i, 3})$ of $\varphi_{4,\ell}$, $S$ contains $q_{\ell,j}$ to imply $y'_{i,j}$ by variables in $S\setminus\{y'_{i,j}\}$.
Therefore, $S$ contains $e_{\ell}$ by a similar argument of Case~4 in (i).
By repeating these arguments, all variables in $E_Y$ are contained in $S$.
This contradicts the assumption.
Hence, all $e_{i}\in E_Y$ are contained in $S$, and this completes the proof.
\end{proof}

\begin{claim}
\label{clm:instance_property2}
Let $S$ be an arbitrary non-empty maximal self-implicating set of $\varphi_\phi$.
Then, for each $j\ (1\le j\le n)$, $S$ contains at least one $q_{i,j}\in Q \ (1\le i\le m)$ if and only if $S$ contains all variables appearing in $\varphi_{3,j}$.
\end{claim}

\begin{proof}
We show the following two propositions, respectively.
\begin{enumerate}[(i)]
    \item For each $j$, if $S$ contains at least one $q_{i,j}$, then $S$ contains all variables appearing in $\varphi_{3,j}$.
    \item For each $j$, if $S$ does not contain any $q_{i,j}$ for all $i$, then $S$ does not contain any variables appearing in $\varphi_{3,j}$.
\end{enumerate}
\subparagraph{Case (i)} From Claim~\ref{clm:instance_property1}, $S$ contains all $e_{i}$ and all $p_{i,j}$.
Since $\varphi_{2,i}$ has the clause $(\overline{p_{i,j}}\vee \overline{x_{j,a}}\vee q_{i,j})$ and $S$ contains at least one $q_{i,j}$ and all $p_{i,j}$, we can show that $S$ contains $x_{j,a}$ to imply $q_{i,j}\in S$ by some variables in $S\setminus\{q_{i,j}\}$. 
Therefore, $S$ contains all variables which appear in $\varphi_{3,j}$ since the property~(a) in Claim~\ref{clm:x_gadget_property} holds.

\subparagraph{Case (ii)} From Claim~\ref{clm:instance_property1}, $S$ contains all $e_i$ and all $p_{i,j}$.
For any $j$ satisfying that $q_{i,j}\notin S$ for all $i$, $S$ does not contain some $x_{j,a}$ from the clause $(\overline{p_{i,j}}\vee \overline{x_{j,a}}\vee q_{i,j})$ in $\varphi_{2,j}$.
We now assume that $S$ contains at least one variable, other than $x_{j,a}\notin S$, that appears in $\varphi_{3,j}$.
The formula $\varphi_{3,j}$ has the sets of variables $D_X, X_\phi, X'_\phi$, and $U$.
Thus, we consider the following four cases depending on which variable sets in $\varphi_{3,j}$ the non-empty maximal self-implicating set $S$ contains:
\begin{enumerate}[(a)]
\item $S$ has at least one variable in $X_\phi\setminus\{x_{j,a}\}$,
\item $S$ has at least one variable in $D_X$,
\item $S$ has at least one in $U$,
\item $S$ has at least one in $X'_\phi$.
\end{enumerate}
\begin{enumerate}[(a)]
\item If $S$ contains some $x_{j,\ell}\in X_\phi$ other than $x_{j,a}$, then $S$ contains all variables appearing in $\varphi_{3,j}$ from the property~(a) in Claim~\ref{clm:x_gadget_property}.
This contradicts the fact that $S$ does not contain $x_{j,a}$.
\item If $S$ contains $d_j\in D_X$, then $S$ contains some $x_{j,k}$ from the property~(b) in Claim~\ref{clm:x_gadget_property}.
This means that $S$ contains all variables appearing in $\varphi_{3,j}$ from the property~(a) in Claim~\ref{clm:x_gadget_property}.
This contradicts the fact that $S$ does not contain $x_{j,a}$.
\item In this case, we consider two cases depending on which variables in $U$ are contained in $S$: (1) $S$ contains $u_{j,3}$ or $u_{j,4}$ (2) $S$ contains $u_{j,1}$ or $u_{j,2}$.

\subparagraph{Case (1)} 
If $S$ contains $u_{j,3}$ or $u_{j,4}$, then $S$ contains $d_j$ to imply $u_{j,3}, u_{j,4}$ from the clauses $(u_{j,3}\vee\overline{d_j}), (u_{j,4}\vee\overline{d_j})$.
Therefore, we can consider this case as a subcase of Case~(b), and this leads to a contradiction by a similar argument to Case~(b).

\subparagraph{Case (2)}
In this case, $S$ contains at least one $x'_{j,k}$ to imply $u_{j,1}, u_{j,2}$ from the property~(c) in Claim~\ref{clm:x_gadget_property}.
Then, $S$ contains $q_{\ell,j}$ to imply $x'_{j,k}\in S$ since positive literal $x'_{j,k}$ appears only in $(\overline{q_{\ell,j}}\vee x'_{j,k})$ of $\varphi_{2,\ell}$.
By similar argument and the clause $(\overline{p_{\ell, j}}\vee\overline{x_{j,k'}}\vee q_{\ell, j})$ of $\varphi_{2,\ell}$, $S$ contains $x_{j,k'}$.
This means that $S$ contains all variables in $\varphi_{3,j}$ from the property~(a) in Claim~\ref{clm:x_gadget_property}.
This contradicts the fact that $S$ does not contain $x_{j,a}$.
\item We can consider this case as a subcase of Case~(c).
Therefore, we can lead the contradiction by a similar argument of Case~(c).
\end{enumerate}
\end{proof}

We now show two auxiliary Lemmas to prove the $\mathsf{NP}$-hardness of {\sc DisConn 3-Horn-3}.

\begin{lemma}\label{lem:nae_to_3-horn-3}
    $\varphi_\phi$ has a locally minimal non-zero satisfying assignment if $\phi$ has an NAE-satisfying assignment.
\end{lemma}

\begin{proof}
    Let $\alpha$ be an NAE-satisfying assignment of $\phi$ and let us construct a satisfying assignment $\alpha'$ of $\varphi_\phi$ from $\alpha$.
    The formula $\varphi_{3,j}$ has the sets of variables $D_X, X_\phi, X'_\phi$ and $U$.
    The formula $\varphi_{4,i}$ has the sets of variables $E_Y, Y, Y'$ and $V$.
    The remaining sets of variables in $\varphi_\phi$ are $P$ and $Q$.
    Then, for each $j$, we set the value of all variables appearing in $\varphi_{3,j}$ to the same value of $\alpha(x_j)$ in $\alpha'$.
    For each $j$, we also set the value of each variable $q_{i,j}\in Q$ and each variable $y'_{i,k}\in Y'$ to the same value of $\alpha(x_j)$, where $y'_{i,k}$ shares the clause $(\overline{q_{i,j}}\vee y'_{i,k})$ with $q_{i,j}$.
    Moreover, we set the value of all variables in $E_Y, Y$, and $P$ to 1 in $\alpha'$.
    For each variable in $V$, we consider the following assignment.
    We set all $v_{i,2}, v_{i,3}$, and $v_{i,4}$ to the value 1 in $\alpha'$.
    For each $v_{i,1}$, we set it to the value 1 if $y'_{i,1}$ or $y'_{i,2}$ is assigned the value 1, and 0 otherwise.
    Then, $\alpha'$ is a satisfying assignment of $\varphi_\phi$.
    For each variable, there exists an implication clause such that if the variable is assigned the value 1 and appears as a positive literal, then all remaining variables in the clause are also assigned the value 1.
    By the definition of the implication clause, any assignment obtained by flipping any 1's bit in $\alpha'$ is not a satisfying assignment of $\varphi_\phi$.
    This means that $\alpha'$ is a locally minimal non-zero satisfying assignment of $\varphi_\phi$.
\end{proof}

\begin{lemma}\label{lem:3-horn-3_to_nae}
    $\phi$ has an NAE-satisfying assignment if $\varphi_\phi$ has a locally minimal non-zero satisfying assignment.
\end{lemma}

\begin{proof}
    Let $\alpha$ be a locally minimal non-zero satisfying assignment of $\varphi_\phi$ and let $S$ be a maximal self-implicating set corresponding to $\alpha$.
    Note that $S$ is the set of variables assigned the value 1 in $\alpha$.
    We now construct a satisfying assignment $\alpha'$ of $\phi$ from $\alpha$ by setting $\alpha(x_j)$ to 1 if $d_j\in S$ and $\alpha(x_j)$ to 0 if $d_j\notin S$ for each $j\ (1\le j\le n)$, and show that $\alpha'$ is an NAE-satisfying assignment.
    From Claim~\ref{clm:instance_property1}, $S$ contains all $e_i\in E_Y\ (1\le i\le m)$ and all $p_{i,j}\in P$.
    This means that $S$ contains at least one $y'_{i,j'}\ (1\le j' \le 3)$ for each $i$ from the property~(d) in Claim~\ref{clm:y_gadget_property}.
    Moreover, for each $i$, $S$ contains at least one $q_{i,j}$ to imply all $y'_{i,j'}\in S$ by variables in $S\setminus\{y'_{i,j'}\}$ since $S$ is a maximal self-implicating set and each $\varphi_{2,i}$ has the clauses $(\overline{q_{i,j}} \vee y_{(i\bmod m)+1,1}'),(\overline{q_{i,k}} \vee y_{(i\bmod m)+1, 2}')$, and $ (\overline{q_{i,\ell}} \vee y_{(i\bmod m)+1, 3}')$.
    From Claim~\ref{clm:instance_property2}, for each $j$ such that $q_{i,j}\in S$, $S$ contains all variables appearing in $\varphi_{3,j}$.
    This means that $i$-th clause $C_i$ of $\phi$, which has $x_{j,a}, x_{k,b}$ and $x_{\ell,c}$, has at least one literal assigned the value 1.
    If $S$ contains all $q_{i,j}$ for some $i$, then $S$ contains $x_{j,a}, x_{k,b}$ and $x_{\ell,c}$ to imply all $q_{i,j}$ from the clauses $(\overline{p_{i,j}}\vee \overline{x_{j,a}}\vee q_{i,j}), (\overline{p_{i,k}}\vee \overline{x_{k,b}}\vee q_{i,k})$, and $(\overline{p_{i,\ell}}\vee \overline{x_{\ell,c}}\vee q_{i,\ell})$ of $\varphi_{2,i}$.
    It means that the value of $x_{j,a}, x_{k,b}$ and $x_{\ell,c}$ is 1, and the clause $(\overline{x_{j,a}}\vee\overline{x_{k,b}}\vee\overline{x_{\ell, c}})$ of $\varphi_1$ is falsified.
    Thus, at least one $x_{j,a}$ must be excluded from $S$, and then for each $i$, at least one $q_{i,j}$ also must be excluded from $S$.
    This means that for each $i$, $i$-th clause $C_i$ of $\phi$ has at least one literal assigned the value 0.
    Therefore, for each $j$ such that $q_{i,j}\notin S$, $S$ does not contain all variables appearing in $\varphi_{3,j}$ from Claim~\ref{clm:instance_property2}.
    For each $j$, we now set $\alpha'(x_j)=1$ if $d_j\in S$, and set $\alpha'(x_j)=0$ if $d_j\notin S$.
    Then, $\alpha'$ is an NAE-satisfying assignment of $\phi$ since each clause in $\phi$ has at least one literal assigned the value 0 and at least one literal assigned the value 1 in $\alpha'$.
    This completes the proof.
\end{proof}

We are ready to prove the following theorem.

\begin{theorem}\label{thm:3-horn-3_npc}
    {\sc DisConn 3-Horn-3} is $\mathsf{NP}$-complete.
\end{theorem}

\begin{proof}
Let $\phi$ be an instance of {\sc Monotone NAE E3-SAT-E4}. 
Let $\varphi_\phi$ be a 3-Horn formula constructed from $\phi$, and $\Phi_{\varphi_\phi}$ be the formula constructed from $\varphi_\phi$. 
By Lemma~\ref{lem:makino_formula}, if $\Phi_{\varphi_\phi}$ has a non-zero satisfying assignment, then the solution graph $G_{\varphi_\phi}$ is disconnected; i.e., a non-zero satisfying assignment of $\Phi_{\varphi_\phi}$ is a certificate of disconnection for $G_{\varphi_\phi}$.
Thus, {\sc DisConn 3-Horn-4} belongs to $\mathsf{NP}$. 
Lemma~\ref{lem:nae_to_3-horn-3}, Lemma~\ref{lem:3-horn-3_to_nae}, and Lemma~\ref{lem:makino_formula} imply that {\sc DisConn 3-Horn-3} is $\mathsf{NP}$-hard. 
Therefore, {\sc DisConn 3-Horn-3} is $\mathsf{NP}$-complete.
\end{proof}

The $\mathsf{coNP}$-completeness of {\sc Conn 3-Horn-3} immediately follows from Theorem~\ref{thm:3-horn-3_npc}.
We then show the following Theorem.

\begin{theorem}\label{thm:3-horn-e3_npc}
    {\sc Conn 3-Horn-E3} is $\mathsf{coNP}$-complete.
\end{theorem}

\begin{proof}
    We prove the $\mathsf{coNP}$-hardness of {\sc Conn 3-Horn-E3} by a polynomial-time reduction from {\sc Conn 3-Horn-3}.
    We first show the construction of a 3-Horn formula with each variable appearing exactly three times from an instance of {\sc Conn 3-Horn-3}.
    Let $\phi$ be an instance of {\sc Conn 3-Horn-3}.
    Then, for every variable $v$ which appears twice, we introduce new variables $v_1, v_2$ and add three clauses $(\overline{v_1}\vee\overline{v_2})(\overline{v_1}\vee\overline{v_2})(\overline{v_1}\vee\overline{v_2}\vee\overline{v})$.
    For each variable that appears once, we repeat the same operation until each variable appears three times.
    Thus, we can obtain the formula $\varphi$, which is an instance of {\sc Conn 3-Horn-E3}.
    Then, we transform $\phi$ and $\varphi$ to $\Phi_\phi$ and $\Phi_\varphi$ for simplicity of proving the correctness of our reduction.
    From the construction, $\Phi_\varphi$ has unit clauses with negative literals of all added variables.
    Therefore, $\Phi_\varphi$ is equal to $\Phi_\phi$ by assigning 0 to all added variables of $\varphi$.
    It means that $\Phi_\phi$ has a non-zero satisfying assignment if and only if $\Phi_\varphi$ has a non-zero satisfying assignment, and by Lemma~\ref{lem:makino_formula} this completes the proof.
\end{proof}

By duality, $\mathsf{coNP}$-completeness of the Boolean connectivity for dual-Horn formulas with the same restrictions follows from the same discussion.

\section{Conclusion}
We proved that {\sc Conn $3$-Horn-E3} is $\mathsf{coNP}$-complete.
Furthermore, we presented a polynomial-time algorithm for {\sc Conn $3$-Horn-2}.
As a result, we established the complexity boundary between $\mathsf{P}$ and $\mathsf{coNP}$-completeness for {\sc Conn $3$-Horn} with bounded variable occurrences.
However, this complexity boundary could not be determined in the case where each clause has exactly three literals.
Although we derived a polynomial-time algorithm for {\sc Conn E$3$-Horn-3}, it remains open whether {\sc Conn E3-Horn-E4} — namely, {\sc Conn $3$-Horn} in which each variable appears exactly four times and each clause has exactly three literals — is $\mathsf{coNP}$-complete.
We conjecture that {\sc Conn E3-Horn-E4} is $\mathsf{coNP}$-complete, since {\sc Monotone NAE 3-SAT} remains $\mathsf{NP}$-complete even when each clause has exactly three literals and each variable appears exactly four times~\cite{Darmann20}, and because the currently known complexity results for {\sc Conn $3$-Horn} with bounded variable occurrences mirror those for {\sc Monotone NAE 3-SAT} under the same restrictions.

\bibliographystyle{plainurl}
\bibliography{ref}

\end{document}